\documentclass[11pt]{amsart}
\reversemarginpar
\usepackage{amsmath}
\usepackage{amssymb}
\usepackage{graphicx}
\usepackage{graphicx}
\usepackage{dcolumn}
\usepackage{bm}
\usepackage{amsfonts}
\usepackage{latexsym}

\begin{document}
\newcommand{\commentout}[1]{}

\newcommand{\nwc}{\newcommand}
\newcommand{\bz}{{\mathbf z}}
\newcommand{\sqk}{\sqrt{\ks}}
\newcommand{\sqkone}{\sqrt{|\ks_1|}}
\newcommand{\sqktwo}{\sqrt{|\ks_2|}}
\newcommand{\invsqkone}{|\ks_1|^{-1/2}}
\newcommand{\invsqktwo}{|\ks_2|^{-1/2}}
\newcommand{\partz}{\frac{\partial}{\partial z}}
\newcommand{\grady}{\nabla_{\by}}
\newcommand{\gradp}{\nabla_{\bp}}
\newcommand{\gradx}{\nabla_{\bx}}
\newcommand{\invf}{\cF^{-1}_2}
\newcommand{\myphi}{\Phi_{(\eta,\rho)}}
\newcommand{\minrg}{|\min{(\rho,\gamma^{-1})}|}
\newcommand{\al}{\alpha}
\newcommand{\xvec}{\vec{\mathbf x}}
\newcommand{\kvec}{{\vec{\mathbf k}}}
\newcommand{\lt}{\left}
\newcommand{\ksq}{\sqrt{\ks}}
\newcommand{\rt}{\right}
\newcommand{\ga}{\gamma}
\newcommand{\vas}{\varepsilon}
\newcommand{\lan}{\left\langle}
\newcommand{\ran}{\right\rangle}
\newcommand{\tvas}{{W_z^\vas}}
\newcommand{\psiep}{{W_z^\vas}}
\newcommand{\wep}{{W^\vas}}
\newcommand{\weptil}{{\tilde{W}^\vas}}
\newcommand{\wepz}{{W_z^\vas}}
\newcommand{\weps}{{W_s^\ep}}
\newcommand{\wepsp}{{W_s^{\ep'}}}
\newcommand{\wepzp}{{W_z^{\vas'}}}
\newcommand{\wepztil}{{\tilde{W}_z^\vas}}
\newcommand{\vvas}{{\tilde{\ml L}_z^\vas}}
\newcommand{\veptil}{{\tilde{\ml L}_z^\vas}}
\newcommand{\vep}{{{ V}_z^\vas}}
\newcommand{\cvc}{{{\ml L}^{\ep*}_z}}
\newcommand{\cvcp}{{{\ml L}^{\ep*'}_z}}
\newcommand{\cvp}{{{\ml L}^{\ep*'}_z}}
\newcommand{\cvtil}{{\tilde{\ml L}^{\ep*}_z}}
\newcommand{\cvtilp}{{\tilde{\ml L}^{\ep*'}_z}}
\newcommand{\vtil}{{\tilde{V}^\ep_z}}
\newcommand{\ktil}{\tilde{K}}
\newcommand{\n}{\nabla}
\newcommand{\tkappa}{\tilde\kappa}
\newcommand{\ks}{{\omega}}
\newcommand{\bx}{\mb x}
\newcommand{\br}{\mb r}
\nwc{\bH}{{\mb H}}
\newcommand{\bu}{\mathbf u}
\nwc{\bxp}{{{\mathbf x}}}
\nwc{\byp}{{{\mathbf y}}}
\newcommand{\bD}{\mathbf D}
\nwc{\bh}{\mathbf h}
\newcommand{\bB}{\mathbf B}
\newcommand{\bC}{\mathbf C}
\nwc{\cO}{\mathcal  O}
\newcommand{\bp}{\mathbf p}
\newcommand{\bq}{\mathbf q}
\newcommand{\by}{\mathbf y}
\nwc{\bP}{\mathbf P}
\nwc{\bs}{\mathbf s}
\nwc{\bX}{\mathbf X}
\newcommand{\pdg}{\bp\cdot\nabla}
\newcommand{\pdgx}{\bp\cdot\nabla_\bx}
\newcommand{\one}{1\hspace{-4.4pt}1}
\newcommand{\corr}{r_{\eta,\rho}}
\newcommand{\rinf}{r_{\eta,\infty}}
\newcommand{\rzero}{r_{0,\rho}}
\newcommand{\rzeroinf}{r_{0,\infty}}
\nwc{\om}{\omega}

\nwc{\nwt}{\newtheorem}
\nwc{\xp}{{x^{\perp}}}
\nwc{\yp}{{y^{\perp}}}
\nwt{remark}{Remark}
\nwt{definition}{Definition} 

\nwc{\ba}{{\mb a}}
\nwc{\bal}{\begin{align}}
\nwc{\be}{\begin{equation}}
\nwc{\ben}{\begin{equation*}}
\nwc{\bea}{\begin{eqnarray}}
\nwc{\beq}{\begin{eqnarray}}
\nwc{\bean}{\begin{eqnarray*}}
\nwc{\beqn}{\begin{eqnarray*}}
\nwc{\beqast}{\begin{eqnarray*}}

\nwc{\eal}{\end{align}}
\nwc{\ee}{\end{equation}}
\nwc{\een}{\end{equation*}}
\nwc{\eea}{\end{eqnarray}}
\nwc{\eeq}{\end{eqnarray}}
\nwc{\eean}{\end{eqnarray*}}
\nwc{\eeqn}{\end{eqnarray*}}
\nwc{\eeqast}{\end{eqnarray*}}

\nwc{\ep}{\varepsilon}
\nwc{\eps}{\varepsilon}
\nwc{\ept}{\epsilon}
\nwc{\vrho}{\varrho}
\nwc{\orho}{\bar\varrho}
\nwc{\ou}{\bar u}
\nwc{\vpsi}{\varpsi}
\nwc{\lamb}{\lambda}
\nwc{\Var}{{\rm Var}}

\nwt{cor}{Corollary}
\nwt{proposition}{Proposition}
\nwt{theorem}{Theorem}
\nwt{summary}{Summary}
\nwc{\nn}{\nonumber}
\nwc{\mf}{\mathbf}
\nwc{\mb}{\mathbf}
\nwc{\ml}{\mathcal}
\nwc{\bd}{{\mb d}}
\nwc{\IA}{\mathbb{A}} 
\nwc{\bi}{\mathbf i}
\nwc{\bo}{\mathbf o}
\nwc{\IB}{\mathbb{B}}
\nwc{\IC}{\mathbb{C}} 
\nwc{\ID}{\mathbb{D}} 
\nwc{\IM}{\mathbb{M}} 
\nwc{\IP}{\mathbb{P}} 
\nwc{\bI}{\mathbf{I}} 
\nwc{\IE}{\mathbb{E}} 
\nwc{\IF}{\mathbb{F}} 
\nwc{\IG}{\mathbb{G}} 
\nwc{\IN}{\mathbb{N}} 
\nwc{\IQ}{\mathbb{Q}} 
\nwc{\IR}{\mathbb{R}} 
\nwc{\IT}{\mathbb{T}} 
\nwc{\IZ}{\mathbb{Z}} 

\nwc{\cE}{{\ml E}}
\nwc{\cP}{{\ml P}}
\nwc{\cQ}{{\ml Q}}
\nwc{\cL}{{\ml L}}
\nwc{\cX}{{\ml X}}
\nwc{\cW}{{\ml W}}
\nwc{\cZ}{{\ml Z}}
\nwc{\cR}{{\ml R}}
\nwc{\cV}{{\ml L}}
\nwc{\cT}{{\ml T}}
\nwc{\crV}{{\ml L}_{(\delta,\rho)}}
\nwc{\cC}{{\ml C}}
\nwc{\cA}{{\ml A}}
\nwc{\cK}{{\ml K}}
\nwc{\cB}{{\ml B}}
\nwc{\cD}{{\ml D}}
\nwc{\cF}{{\ml F}}
\nwc{\cS}{{\ml S}}
\nwc{\cM}{{\ml M}}
\nwc{\cG}{{\ml G}}
\nwc{\cH}{{\ml H}}
\nwc{\bk}{{\mb k}}
\nwc{\bT}{{\mb T}}
\nwc{\cbz}{\overline{\cB}_z}
\nwc{\supp}{{\hbox{supp}(\theta)}}
\nwc{\fR}{\mathfrak{R}}
\nwc{\bY}{\mathbf Y}
\newcommand{\mbr}{\mb r}
\nwc{\pft}{\cF^{-1}_2}
\nwc{\bU}{{\mb U}}
\nwc{\bPhi}{{\mb \Phi}}
\title{ Compressive Imaging of Subwavelength Structures}
\author{Albert C.  Fannjiang}
\email{
fannjiang@math.ucdavis.edu}
  \thanks{The research is partially supported by
the NSF grant DMS - 0908535.}

      \address{
   Department of Mathematics,
    University of California, Davis, CA 95616-8633}
   
       \begin{abstract}
The problem of imaging extended targets  (sources or
scatterers) is formulated in the framework of compressed
sensing with emphasis  on subwavelength
resolution.

 The  proposed formulation of the
 problems of inverse source/scattering   
 is essentially exact and leads to the 
 random partial Fourier measurement matrix. 
 In the case of square-integrable targets, 
 the proposed  sampling scheme in the Littlewood-Paley wavelet basis  block-diagonalizes 
 the scattering matrix 
with each block 
 in the form of random partial Fourier matrix
 corresponding to each dyadic scale of the target. 
 
The  resolution issue is analyzed from two perspectives:
stability and the signal-to-noise ratio (SNR). 
The  subwavelength modes are shown to be  
 typically unstable. The  stability in the subwavelength
 modes requires  additional techniques such as near-field measurement or illumination. The number of the
 stable modes typically increases as the negative $d$-th  (the dimension of the target)
 power of the distance between the target and the
 sensors/source.
 The resolution limit is shown to be  inversely proportional to
the  SNR in the high SNR limit. 

  Numerical simulations are provided to validate  the theoretical predictions. 

       \end{abstract}
       
       \maketitle

\section{Introduction}

The  Compressed Sensing paradigm 
is supplying 
 a fresh perspective on the imaging problems, including 
 source inversion and inverse scattering (see the extensive
 literature in \cite{CSR}) and we aim to analyze these
 classical problems from the perspective of compressed sensing. 
 
 The main purposes of the  paper are  two folds: 
 (i) to formulate the problems 
 of  inverse source and scattering for  extended targets  
 in the framework of compressed sensing
 and analyze  it by the fundamental results 
 on the random Fourier measurements
\cite{CRT1, CRT2, CT, CT2, Rau,  RV}
and 
(ii) to use the compressed sensing solution thus obtained 
as a starting point to investigate the problem of  subwavelength resolution  in the presence of noise,  in particular,
to discuss the idea of extracting
 subwavelength information by
 near-field measurement and illumination,
 common in nano-optics \cite{NH}.

 The important feature of our  formulation of the
 imaging problems  
 is that it is essentially exact (without the paraxial
 approximation as used in \cite{CS-par}) and
 is for extended (periodic or non-periodic) targets.
 As such, our formulation provides
 an instructive  example for assessing the
 power of the compressed sensing techniques 
 when applied to the physical problems of imaging. 
 In the case of
  inverse scattering for square-integrable targets, 
  we use the Littlewood-Paley wavelet basis
  and propose a sampling scheme by which 
 the scattering matrix can be  block-diagonalized
with each block 
 in the form of random partial Fourier matrix
 corresponding to each dyadic scale present
 in the target structure. By this approach
 we can image the extended target scale-by-scale.

We analyze the  resolution issue from two perspectives.
First we define the stably recoverable (or stable for short) modes
in terms of the noise amplification factor.
 We show that  the  subwavelength modes are  
 typically unstable. To achieve stability in the subwavelength
reconstruction  we need additional techniques such as near-field measurement or illumination with which the number of the
 stable modes typically increases as the negative $d$-th  (the dimension of the target)
 power of the distance between the target and the
 sensors/source. Second, we 
 take into account the signal-to-noise ratio (SNR) in the resolution analysis and show that
 the resolution limit is inversely proportional to
the  SNR in the high SNR limit. 
 
The paper is organized as follows. 
We treat  the case of periodic targets using the Fourier
basis, first
for the inverse scattering in Section \ref{sec2} and then for the
inverse (Born)  scattering problem in Section \ref{sec3}. We present numerical
 results confirming the theoretical predictions in Section \ref{sec5}. In Section \ref{sec4}
 we analyze  the case of square-integrable targets
 using the Littlewood-Paley basis and point out the extension
 to three dimensions. 
 We conclude in Section \ref{sec6}.

  \section{Source inversion}
\label{sec2}
 
First  we consider imaging of an   extended source (the target)  in two dimensional
 $(z,x)$-plane. 
 We assume that the target is  located
 in the line $z=z_0$. 
 
 The  target
 is represented by the variable source  amplitude $\sigma(x)$ which is assumed
 to be 
 periodic with period $L$ and admits  
 the Fourier expansion 
 \beq
 \label{0.1}
 \sigma(x)=\sum_{k=-\infty}^{\infty} \hat \sigma_{k}
 e^{i2\pi kx/L}. 
 \eeq
 We consider the class of  band-limited functions  with $\hat\sigma_k=0$ except for 
\beq
\label{20}
 k\in [-(m-1)/2,  (m-1)/2]
\eeq
in (\ref{0.1}) where $m$ is an odd integer so that
there are at most $m$ relevant Fourier modes. 

The wave propagation in the free space  is governed  by the Helmholtz equation 
\[
\Delta u+\om^2 u=0. 
\]
where  
 $\om=2\pi/\lambda$ is the frequency assuming the wave speed is unity.  In two dimensions,
the Green function  $G(\br)$ is 
 \[
 G(\br)={i\over 4} H^{(1)}_0(\om  |\br|),\quad \br=(z,x)
 \]
 where $H^{(1)}_0$ is the zeroth order Hankel function
 of the first kind.  $G$ can be expressed by
  the Sommerfeld integral representation
  \beq
  \label{somm}
 G(\br)={i\over 4\pi}
  \int e^{i\om(|z|\beta(\alpha)+x\alpha)} {d\alpha\over \beta(\alpha)},  \eeq
  where 
  \beq
  \label{16}
  \beta(\alpha)=\lt\{\begin{array}{ll}
  \sqrt{1-\alpha^2}, & |\alpha|<1\\
  i\sqrt{\alpha^2-1},& |\alpha|> 1
  \end{array}
  \rt.
  \eeq
  \cite{BW}. The integrand in (\ref{somm}) with  
  real-valued $\beta$ (i.e. $|\alpha|<1$)
corresponds to the homogeneous wave
and that with imaginary-valued $\beta$ (i.e. $|\alpha|>1$) corresponds
to the evanescent (inhomogeneous) wave which has
an exponential-decay factor $e^{-\om |z| \sqrt{\alpha-1}}$. 

The signal arriving at the sensor located at $(0, x)$ is given by
\beq
\label{sense}
\int G(z_0, x-x') \sigma (x') dx'&=&
{i\over 2\om}\sum_{k} {\hat\sigma_k\over \beta_k}  e^{i\om z_0\beta_k } e^{i\om \alpha_k x}
\eeq
by (\ref{somm}) where 
\beq
\label{17}
\alpha_k= {k\lambda\over L}, & & \beta_k = \beta(\alpha_k).
  \eeq
By (\ref{17}), the subwavelength structure of the target  is encoded in
$\hat \sigma_k$  with $|k|\lambda >L$ which are  conveyed  by 
  the evanescent waves. 

Let $(0, x_j), x_j=\xi_jL, j=1,...,n$ be the coordinates
of the $n$ sensors in the line $z=0$. 
To set the problem in the framework of compressed sensing 
 we 
 define the signals received by the $n$ sensors as 
 the measurement vector $Y$ and  set the target vector
 $X=(X_k)\in\IC^m$ as 
 \beq
 \label{13}
 X_k={i \sqrt{n} e^{i\om z_0\beta_k }\over 2\om \beta_k}   
 {\hat\sigma_k}.
 \eeq
  To avoid a vanishing denominator in (\ref{13}), we must set $L/\lambda\not \in \IN$ so that $\beta_k\neq 0$. 
This leads to the form
\beq
\label{u1}
Y=\bPhi X
\eeq
where  
the sensing matrix  $\bPhi =[\Phi_{jk}] \in \IC^{n\times m}$ 
has the entries
\beq
\label{0.6}
\Phi_{jk}&=&{1\over \sqrt{n}} e^{i\om\alpha_k x_j}={1\over \sqrt{n}} e^{i2\pi k \xi_j},\quad
j=1,...,n,\quad k=1,...,m.
\eeq
\commentout{
 \beq
 \label{1.10}
 A_{jk}&=&\lt\{\begin{array}{ll}
 {e^{i\om z_0 \beta_k}\over \sqrt{n}} e^{i\om \alpha_k \xi_j} ,
 &\mbox{for} \quad \alpha_k\leq 1\\
 {-i\over \sqrt{n}} e^{i\om \alpha_k \xi_j} ,
 &\mbox{for}\quad  \alpha_k > 1
 \end{array}\rt.
 \eeq
 for $
\forall j=1,...,n$.
}
Hence the imaging procedure is split into two stages:
the first is to ``invert''  $\Phi$ and solve for $X$, and
the second is to deduce $\sigma$ from $X$.
The sensing matrix (\ref{0.6}) is designed  to mimic 
the random partial Fourier matrix which has
the restricted isometry property and is amenable
to the compressed sensing techniques (see below). 
Note that the formulation is essentially exact and 
no far-field approximation is  made. 

 The main thrust of compressive sensing is that under
suitable conditions the  inversion
can  be  achieved as the $\ell^1$-minimization
\beqn
\mbox{min} \|X\|_1,\quad\mbox{subject to}\quad Y=\bPhi X
\eeqn
which also goes by the name of {\em Basis Pursuit} (BP)
and can be solved by linear programming. 
BP was first introduced empirically in seismology by Claerbout and Muir and later studied mathematically by Donoho and
others  \cite{BDE, CDS, DH}.

\commentout{
This can be achieved by proper normalization of the vectors noting  that the amplitude in (\ref{sense}) 
\beq
\label{kth}
-{i\over 2\om \beta_k}  e^{i\om z_0\beta_k } 
\eeq
is independent of the row-index  $j$. This amounts to
considering the target vector  $X=(X_k)= \bT \hat\sigma $
where $\hat\sigma=(\hat\sigma_k) \in \IC^m$ and
$\bT$ is the diagonal matrix with $k-$th entry given by
$\sqrt{m}$ times the quantity  (\ref{kth}). 
}

A fundamental notion in compressed sensing under which
BP yields the unique exact solution is the restrictive isometry property due to Cand\`es and Tao \cite{CT}.
Precisely, let  the sparsity $s$  of the target vector be the
number of nonzero components of $X$ and define the restricted isometry constant $\delta_s$
to be the smallest positive number such that the inequality
\[
(1-\delta_s) \|Z\|_2^2\leq \|\bPhi Z\|_2^2\leq (1+\delta_s)
\|Z\|_2^2
\]
holds for all $Z\in \IC^m$ of sparsity at most $ s$. 
\commentout{
\[
(1-\delta_s) \|Z\|_2^2\leq \|\bPhi_S Z\|_2^2\leq (1+\delta_s)
\|Z\|_2^2
\]
holds for all $Z\in \IC^{|S|}$ and all subsets $S\subset \{1,...,m\}$ of size
$|S|\leq s$ where $\bPhi_S$ denotes the $p\times  |S|$ matrix
consisting of the columns of $\bPhi$ indexed by $S$. 
}

For the target vector $X$ let $X_s$ denote the best $s$-sparse
approximation of $X$ in the sense of $L^1$-norm, i.e.
\[
X^s=\hbox{argmin}\,\, \|Z-X\|_1,\quad\hbox{s.t.}
\quad \|Z\|_0\leq s
\]
where $\|Z\|_0$ denotes the number of nonzero components, called the sparsity,  of $Z$. 
Clearly, $X^s$ consists of the $s$ largest components
of $X$.

Now we state the fundamental   result 
of the RIP approach \cite{Can} which is an improvement
of the results of \cite{CRT1, CT}. 
\begin{theorem} \label{thm1} \cite{Can}
Suppose the  restricted isometry constant of $\bPhi$
satisfies the inequality 
\beq
\label{ric}
\delta_{2s}<\sqrt{2}-1
\eeq
 Then the solution $X_*$ by BP satisfies
 \beq
 \label{100}
 \|X_*-X\|_1&\leq & C_0\|X-X^s\|_1\\
 \|X_*-X\|_2&\leq & C_0s^{-1/2}\|X-X^s\|_1\label{101}
 \eeq
 for some constant $C_0$. In particular,
 if $X$ is $s$-sparse, then the recovery is exact.
\end{theorem}
\begin{remark}
Greedy algorithms have significantly lower  computational
complexity than linear programming and have 
provable performance under various conditions.
For example
under the condition  $\delta_{3s}<0.06$
 the Subspace Pursuit (SP) algorithm is guaranteed to exactly recover $X$ via a finite number of iterations \cite{DM}.
 See \cite{NT} for a  closely related algorithm (CoSaMP).

\end{remark}

\begin{theorem}\cite{Rau}\label{thm2}
Let $\xi_j, j=1,...,n$ be independent, uniform random variables in $[0,1]$. 
Suppose \beq
\label{77}
{n\over \ln{n}}\geq C \delta^{-2} r \ln^2{r} \ln{m} \ln{1\over \eta},\quad
\eta\in (0,1)
\eeq
for a given sparsity $r$ where $C$ is an absolute constant. 
Then the restricted isometry constant 
of  the matrix for random Fourier measurements  satisfies 
\[
\delta_r\leq \delta
\]
 with
probability at least $1-\eta$.
\end{theorem}
See \cite{CRT1, CT2,  RV}  for the case when $\xi_l$ belong to 
the discrete  subset of  $[0,1] $ of equal spacing $1/\sqrt{m}$. 

\commentout{
To realize the restrictive isometry property for our matrix $\bPhi$, we follow \cite{CRT1} and  let $I_1, I_2,..., I_m$
be independent Bernoulli random variables taking the value
1 with probability $\bar n/m, \bar n\ll m.$  Let 
 \beq
 \label{10}
 \xi_j=\eta_j L/m
 \eeq
 where  $\eta_j$ belongs to
the random set
\beq
\Omega= \{ q\in \{0,...,m-1\}| I_q=1\}.
\eeq
We have  $\IE|\Omega|=\bar n$ and indeed by 
Bernstein's inequality $ n=|\Omega|$ is close to $\bar n$
with high probability. 
\commentout{ \cite{AS} that
\beq
\label{bern}
\IP\lt(|\Omega|\geq (1-\sigma) n \rt)\geq 1-e^{-n\sigma^2/2}.
\eeq
}

Now we state two fundamental results essential for our problem. 
\begin{theorem} \cite{CT}\label{thm1}
Suppose the true target vector $X$ has
the sparsity at most $s$. 
Suppose the  restricted isometry constant of $\bPhi$
satisfies the inequality 
\beq
\label{ric}
\delta_{3s}+3\delta_{4s}< 2. 
\eeq
 Then $X$ is the unique solution of BP.
\end{theorem}
\begin{remark}
Greedy algorithms have significantly lower  computational
complexity than linear programming and have 
provable performance under various conditions.
For example
under the condition  $\delta_{3s}<0.06$
 the Subspace Pursuit (SP) algorithm is guaranteed to exactly recover $X$ via a finite number of iterations \cite{DM}.
 
\end{remark}

\begin{theorem}\cite{RV}\label{thm2}
For any $\tau>1$ and any $m, s>2$, a random set $\Omega$ 
of average cardinality 
\beq
\bar n=(C \tau s \log{m}) \log{(C\tau s\log{m})} \log^2 s\label{12}
 \eeq
 gives rise to a sensing matrix $\bPhi$ which 
satisfies the restricted isometry condition (\ref{ric}) 
with probability at least $1-5 e^{-c\tau}$. Here $c$ and $C$
are some absolute constants. 
\end{theorem}
}

\commentout{
\begin{theorem}\cite{CT2}\label{thm2}
For any $\tau>1$ and for all $m$ sufficiently large
(depending on $\tau$), a random set $\Omega$ 
of cardinality 
\beq
\label{10}
|\Omega|\geq C\tau s \log^6 m
\eeq
satisfies the restricted isometry condition (\ref{ric}) 
with probability at least $1-m^{-c\tau }$. Here $c$ and $C$
are some absolute constants. 
\end{theorem}
}
\commentout{
\begin{remark}
By the estimate (\ref{bern}), (\ref{10}) holds
with probability at least $1-e^{-n\sigma^2/2}$ for
\beq
\label{11}
n\geq C\tau s \log^6 m/(1-\sigma)
\eeq
 and so does 
the restrictive isometry condition (\ref{ric}) 
with probability at least 
$(1-e^{-n\sigma^2/2})(1-m^{-c\tau})$. 
\end{remark}
\begin{remark}
Theorem \ref{thm2} in some aspect  has been improved by 
. More precisely, condition (\ref{11}) is replaced  by 
 \beq
 n=C \tau s \log{m} \log{(C\tau s\log{m})} \log^2 s\label{12}
 \eeq
 at the expense of a 
 smaller lower  bound $
 1-5e^{-c\tau}$ for the success probability. 
 \end{remark}
 }
\commentout{
\begin{remark}
Theorem \ref{thm2} in some aspect  has been improved by 
 Rudelson and
 Vershynin \cite{RV}. More precisely, condition (\ref{11}) is replaced  by 
 \beq
 n=C t s \log{m} \log{(Cts\log{m})} \log^2 s\label{12}
 \eeq
 at the expense of an enlarged aperture
 \[
 \xi_j=t\eta_j L/m, \quad t>1, 
 \]
 instead of (\ref{10}), which also controls
 the lower bound for the success probability $
 1-5e^{-ct}$. In other words, to achieve  the reduction of
 the power of $\log{m}$ from 6 in (\ref{11}) to
 $1^+$ in (\ref{12}) with a similar success probability
 bound, the sensor aperture has to be substantially enlarged. 
\end{remark}
}
As a consequence of Theorem \ref{thm1} and \ref{thm2},
the target vector $X$ can be determined by BP and
then the source amplitude $\sigma(x)$ can be reconstructed exactly
from (\ref{13}),
including all the subwavelength structures. In other words,  in the absence
of noise there is essentially
no limitation to the resolving power of the compressed
sensing technique, subwavelength or not, as long
as sufficient number of measurements are made.  

When noise is present, however,  the performance of the above  approach may be severely   limited, especially in the recoverability of
 subwavelength information. Consider  the standard model of additive noise
\beq
\label{31}
Y^\ep=\bPhi X+E
\eeq
where $\|E\|_2=\ep>0$ is the size of the noise and 
the associated relaxation scheme
\beq
\label{32}
\mbox{min} \|X\|_1,\quad \mbox{subject  to}\quad \|Y^\ep-\bPhi X\|_2\leq \epsilon.
\eeq
The next result  is a restatement of the result of \cite{Can} 
after applying  Theorem  \ref{thm2} with $r=2s, \delta<\sqrt{2}-1$. 
\begin{theorem}\label{thm8}
 Let $X^\ep$ be the solution to (\ref{32}).
Then under the assumptions of Theorem (\ref{thm2})  
we have 
 \beq
 \|X^\ep-X\|_2\leq C_0s^{-1/2} \|X-X^s\|_1+ C_1\ep
 \eeq
with
probability at least $1-\eta$
where  $C_0$ and $C_1$ are
constants.   \end{theorem}
See also  \cite{DET, Tro3} for
stability result under the condition of incoherence.

\commentout{
\begin{theorem}\cite{CRT2}\label{thm3}
 Let $X^\ep$ be the solution to (\ref{32}).
 If  $\|X\|_0\leq s$ and (\ref{ric}) is true, then
 \[
 \|X^\ep-X\|_2\leq C_s \ep
 \]
 where $C_s$ depends only on  and behaves reasonably with $\delta_{4s}$. 
 \end{theorem}
}

 Inverting the relationship (\ref{13}) with small error in the
 target vector $X$ 
produces a mildly amplified error for those $\hat \sigma_k$ such that
\beq
\label{52}
|e^{i\om z_0\beta_k}|\geq e^{-2\pi}
\eeq
 but significantly amplified error otherwise.  
Here the transition is not clear-cut, however. The choice of the noise amplification threshold   $e^{2\pi}$
 as the stability criterion is convenient but arbitrary;
 any constant less than one will serve our purpose.
 
  The {\em stable}  (or {\em stably recoverable})  modes are  
those corresponding to  $|\alpha_k|\leq 1$ as well as  $|\alpha_k|>1$ such that
\beq
\label{51}
\om |\beta_k| z_0\leq 2\pi 
\eeq
or equivalently
\beq
\label{resolve}
{k}\leq L \sqrt{\lambda^{-2}+  z_0^{-2}}
\eeq
Hence reducing the distance $z_0$ between the sensor array and
the target can  effectively  enlarge the number of stable modes
and more of the subwavelength modes become stably recoverable
as $z_0$ decreases below wavelength. 
This is the idea behind
 the near-field imaging systems such as  the scanning 
microscopy.

The above reconstruction essentially consists of two
stages: Stage 1 involves  the compressed sensing techniques
which  is  always stable (Theorem \ref{thm8}) and Stage 2 is a simple inversion of the  diagonal matrix
\beq
\label{13'}
\hbox{diag}\lt({i \sqrt{n} e^{i\om z_0\beta_k }\over 2\om \beta_k}\rt), 
\eeq
cf. (\ref{13}), which  is  stable if and only if  (\ref{51}) holds
for all $k$ such that $\hat\sigma_k\neq 0$.
When (\ref{51}) is indeed violated, brute force inversion
of the matrix (\ref{13'}) would lead to an enormous error
and hence  a regularization is called for. In such a case, we can
use the well-known Tikhonov regularization to invert (\ref{13'}).  

\commentout{
Condition (\ref{51}) leads to the inequality
\beqn
\label{51-2}
\alpha_k={k\lambda\over L}\leq \sqrt{1+\lt({\lambda\over z_0}\rt)^2}
\eeqn
}

The analysis from (\ref{52}) to (\ref{resolve}) focuses on
the stability issue based on the  noise amplification factor. 
On the other hand, 
it is well accepted that the resolution of any imaging system
should be a function of  the signal-to-noise-ratio 
(see \cite{RA-JOSA} and references therein). 
Let us now analyze the resolution  from this perspective.

Let $X^\ep=(X^\ep_k)$ be the solution of
the convex relaxation scheme (\ref{32}) and
let 
\beq
\label{13''}
\hat\sigma^\ep_k={2\om \beta_k\over i\sqrt{n}} e^{-i\om z_0\beta_k} X^\ep_k.
\eeq
For simplicity, let us assume that the sparsity of $X$ is at most $s$ and therefore by Theorem \ref{thm8} we have
\beq \label{16'}
 \|X^\ep-X\|_2\leq  C_1\ep
 \eeq
 where the constant $C_1$ is well-behaved as demonstrated
 in the numerical results of Section  \ref{sec5}. 

We say that the Fourier mode $\hat\sigma_k$ is {\em resolved }  
if the following
inequality holds:
\beqn
\lt|{\hat\sigma^\ep_k\over \hat \sigma_k}-1\rt|
\leq {1\over K} 
\eeqn
where the constant $K$ is sufficiently larger than unity.
Using (\ref{16'}) and (\ref{13''}) we obtain 
\beq
\lt|{\hat\sigma^\ep_k\over \hat \sigma_k}-1\rt|
&\leq& {2\om |\beta_k|\over \sqrt{n}} \lt| e^{-i\om z_0\beta_k}\rt|{\lt|X^\ep_k-X_k\rt|\over |\hat\sigma_k|}\nn\\
&\leq&  {2C_1\om |\beta_k|}{\ep\over \lt| e^{i\om z_0\beta_k}\rt| |\hat\sigma_k|\sqrt{n}}.\label{255}
\eeq
The expression
\[
\hbox{SNR}_k={\lt| e^{i\om z_0\beta_k}\rt| |\hat\sigma_k| \sqrt{n}/\ep}
\]  has the meaning of
the {\em signal-to-noise ratio for the $k$-th Fourier mode}
and we shall call it  as such. 

\commentout{
Consider the case of sufficiently large
$\hbox{SNR}_k$
\beq
\label{snr}
  \hbox{SNR}_k\gg 1,\quad \hbox{SNR}_k\gg {\sqrt{n}\over 2\om |\beta_k|C_1}
 \eeq
 for all nonvanishing modess $\hat\sigma_k$ in which case we shall just write $\hbox{SNR}$. 
 } Assuming  that the right hand side of (\ref{255})
is at most $1/K$ we obtain from (\ref{255}) the
inequality
\beq
\label{24'}
\lt|\beta_k\rt|\leq {\hbox{SNR}_k\over 2C_1 K\om}. 
\eeq
and consequently 
\beq
\label{266}
{L\over k} \geq \lt({1\over \lambda^2}+\lt({{\hbox{SNR}_k}\over 4\pi C_1 K}\rt)^2\rt)^{-1/2}. 
\eeq
The resolution limit is then obtained by minimizing the left
hand side of (\ref{266}). For simplicity, considering
the case  
that $\hbox{SNR}_k=\hbox{SNR}$ is independent of $k$ for all
nonzero modes $\hat\sigma_k$ we obtain
the resolution limit 
\beq
\label{resolve'}
\lt({1\over \lambda^2}+\lt({{\hbox{SNR}}\over 4\pi C_1 K}\rt)^2\rt)^{-1/2}
\eeq
which essentially says that
the resolution  is inversely proportional to $\hbox{SNR}$
in the high SNR limit.  

Next   let us discuss a different mechanism of  
 superresolution  available in the context of
inverse scattering.

\section{Inverse scattering}
\label{sec3}

Consider  a periodic scatterer
with scattering amplitude $\sigma$ admitting the representation
(\ref{0.1}). 
For simplicity, we use the Born scattering model \cite{BW}  under which  the scattered field at $z=0$  is given by
 \beq
 \label{15}
 u^{\rm s}(0, x)=\int G(z_0, x-x') \sigma(x') u^{\rm i}(z_0, x') dx'
 \eeq
 where $u^{\rm i}(z,x)$ is the incident field. Now with
 the normally incident plane wave
 $u^{\rm i}(z,x)=e^{i\om z}$, eq. (\ref{15}) is essentially
 reduced to (\ref{sense}). 
 
Consider the obliquely  incident  plane wave  $u^{\rm i}(z,x)=e^{i\om(\alpha x+\beta |z-z_1|)}$ where $\beta$ is related to $\alpha$ as in (\ref{17}) and $z_1$, assumed larger than $z_0$. 
Set 
\[
\alpha={q\lambda\over L},\quad q\in \IR.
\]
In the standard setting the illumination field is
a homogeneous wave with  $|\alpha|< 1$. 
The evanescent illumination $|\alpha|>1$ will also be considered here. In such case, $z_1$ is the $z$-coordinate of the illumination source. 

The same calculation as before now leads to
\beqn
u^{\rm s}(0,x)={i\over 2\om}\sum_{k} {\hat\sigma_k\over \beta_k}  e^{i\om z_0\beta_k} e^{i\om (z_1-z_0)\beta} e^{i\om \alpha_kx} 
\eeqn
where  instead of (\ref{17}) we have
\beqn
\alpha_k= {(k+q) \lambda\over L}, & & \beta_k =
\lt\{\begin{array}{ll}
  \sqrt{1-\alpha_k^2}, & |\alpha_k|<1\\
  i\sqrt{\alpha_k^2-1},& |\alpha_k|> 1
  \end{array}
  \rt.. 
  \eeqn

Define  the target vector
 $X^{(q)}=(X^{(q)}_k)\in\IC^m$ as 
 \beq
 \label{14}
 X^{(q)}_k={i  \sqrt{n} e^{i\om z_0\beta_k} e^{i\om (z_1-z_0)\beta} \over 2\om \beta_k}   
 {\hat\sigma_k}
 \eeq
 and proceed as before. Note that $e^{i\om (z_1-z_0)\beta} $
 is a constant factor determined solely by the illumination source
 (distance and angle). 
  To avoid a vanishing denominator 
 we require
 \beq
 \label{21-4}
 {(k+q) \lambda\over L}\neq 1,\quad\forall k.
 \eeq
 Theorem \ref{thm1} and \ref{thm2}
are applicable  to the shifted Fourier matrix $\bPhi^{(q)}=[ \Phi^{(q)}_{jk}]$ with 
 \[
\Phi^{(q)}_{jk}= {1\over \sqrt{n}} e^{i2\pi  (k+q) \xi_j}
 \]
 where  $\xi_j$ are independent, uniform random variables in $[0,1]$.
 
\commentout{ 
In the presence of small noise, all the modes $\hat\sigma_k$  corresponding to $|\alpha_k|<1$ can still be recovered stably. This means all the modes in the set
\beq
\label{21}
\{k: |k+q|<L/\lambda\}
\eeq
can be recovered stably. 

If we now conduct compressive imaging with
$q=\pm L/\lambda $ and recover {\em only}
the modes satisfying (\ref{21}), then 
we can recover stably the modes in the set
$\{k: |k|<2L/\lambda\}$. In other words, we have two-fold
reduction in the smallest stably resolvable scale by means of the two extreme homogeneous illuminations.
}

The same stability analysis as  in Section \ref{sec2} implies 
that the stable modes in the Born scattering satisfy
the constraint 
\beq
|\alpha_k|={|k+q|\lambda\over L}\leq \sqrt{1+{\lambda^2\over z^2_0}} \label{23-3}
\eeq
\commentout{
either 
\[
|\alpha_k|={|k+q|\lambda\over L}<1,\quad\hbox{for the homogeneous modes}
\]
or (\ref{51-2}) must be satisfied (for the evanescent modes). 
}
which implies that  for each $q$  
\beq
\label{25-5}
-q-\sqrt{{L^2\over\lambda^2}+{L^2\over z^2_0}}\leq k\leq -q+
\sqrt{{L^2\over\lambda^2}+{L^2\over z^2_0}}. 
\eeq
Let us maximize the range of stable modes
 (\ref{25-5})  under the
constraint $|e^{i\om (z_1-z_0)\beta}|>e^{-2\pi}$ by \commentout{
Consider now  the case $z_0\gg \lambda$ (far-field measurement)  so that (\ref{resolve2})
 becomes 
 \beq
\label{resolve22}
 \lt( {|q|\over L}+ {1\over \lambda}
 \rt)^{-1}
\eeq
provided that $|e^{i\om (z_1-z_0)\beta}|>e^{-2\pi}$.
}
considering  two different illumination sources: homogeneous and
evanescent wave sources.  
If the target is illuminated by incident  homogeneous waves with
 $|q|< L/\lambda$, then the stable modes are 
 \beq
\label{25-5'}
|k| \leq {L\over \lambda} +
\sqrt{{L^2\over\lambda^2}+{L^2\over z^2_0}}. 
\eeq
The same resolution analysis as before
leads to
\beq
\lt|{\hat\sigma^\ep_k\over \hat \sigma_k}-1\rt|
&\leq&  {2\om |\beta_k|\over \sqrt{n}} \lt| e^{-i\om z_0\beta_k}\rt| \lt|e^{-i\om(z_1-z_0)\beta}\rt| C_1{\ep\over |\hat\sigma_k|}. \label{256}
\eeq
The appropriate definition of the signal-to-noise ratio in this case is
\[
\hbox{SNR}_k={\sqrt{n}  |\hat\sigma_k|\lt| e^{i\om z_0\beta_k}\rt|\lt|e^{i\om(z_1-z_0)\beta}\rt|/ \ep}
\]
after taking into account the intensity of illumination. 
Setting the right hand side of (\ref{256}) to be at most $1/K$ then leads to  (\ref{24'}) with $\alpha_k$ given by (\ref{23-3}).
  which in turn yields the resolution limit
\beq
\label{resolve22}
\lt({1\over \lambda}+ \sqrt{{1\over \lambda^2}+\lt({\hbox{SNR}\over 4\pi C_1 K}\rt)^2}\rt)^{-1}. 
\eeq

On the other hand, if the incident wave is  evanescent and subject to
the constraint $|e^{i\om (z_1-z_0)\beta}|>e^{-2\pi}$, 
which implies
\beq
\label{222}
{|q|}< \sqrt{{L^2\over \lambda^2} +{L^2\over (z_1-z_0)^2}}, 
\eeq
then the  stable modes  according to (\ref{25-5})
satisfy
 \beq
|k|<\sqrt{{L^2\over \lambda^2} +{L^2\over (z_1-z_0)^2}}  +\sqrt{{L^2\over \lambda^2}+{L^2\over z^2_0}}.\label{26}
\eeq

To maximize the range of stable modes limited  by (\ref{26})
we choose $q_*>0$,  which is  sufficiently close to the right hand
side of (\ref{222}) and satisfies  (\ref{21-4}),  and
illuminate the target by   a few evanescent waves
with $ q\in [-q_*, q_*]$. The union of their
respective stable modes (\ref{25-5}) is
the totality of stable modes. 
In optics, the evanescent illumination  can be produced physically by,  for example, the
total internal reflection \cite{NH}. 

Using the same evanescent illumination procedure we can achieve the
resolution limit
\beq
\label{resolve33}
\lt(\sqrt{{1\over \lambda^2} +{1\over (z_1-z_0)^2}}+ \sqrt{{1\over \lambda^2}+\lt({\hbox{SNR}\over 4\pi C_1 K}\rt)^2}\rt)^{-1}
\eeq
which indicates two ways of improving resolution without
increasing the probe frequency: reducing
the distance between  the plane-wave source and the target,
and increasing the signal-to-noise ratio.

\commentout{
which can be cast into the form (\ref{sense}) as
\beq
-{i\over 2\om\sqrt{m}}\sum_{k'} {\hat\sigma'_{k'-q}\over \beta_{k'-q}}  e^{i\om z_0\beta_{k'} } e^{i\om k' \lambda \xi/L}
\eeq 
with $k'=q-(m-1)/2,..., q+(m-1)/2.$ 
}

\section{Numerical results}\label{sec5}
\begin{figure}
\includegraphics[width=0.4\textwidth]{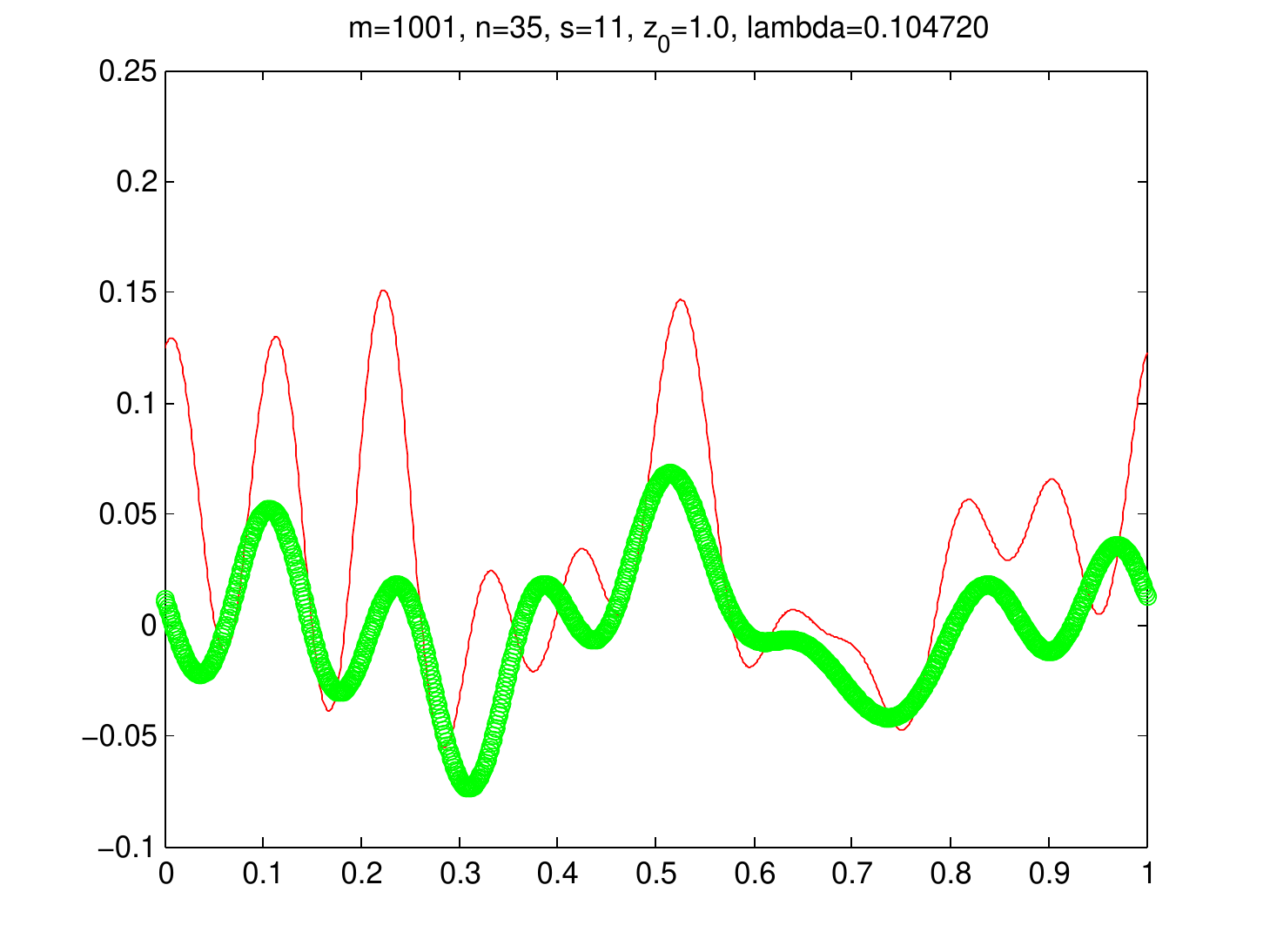}
\includegraphics[width=0.4\textwidth]{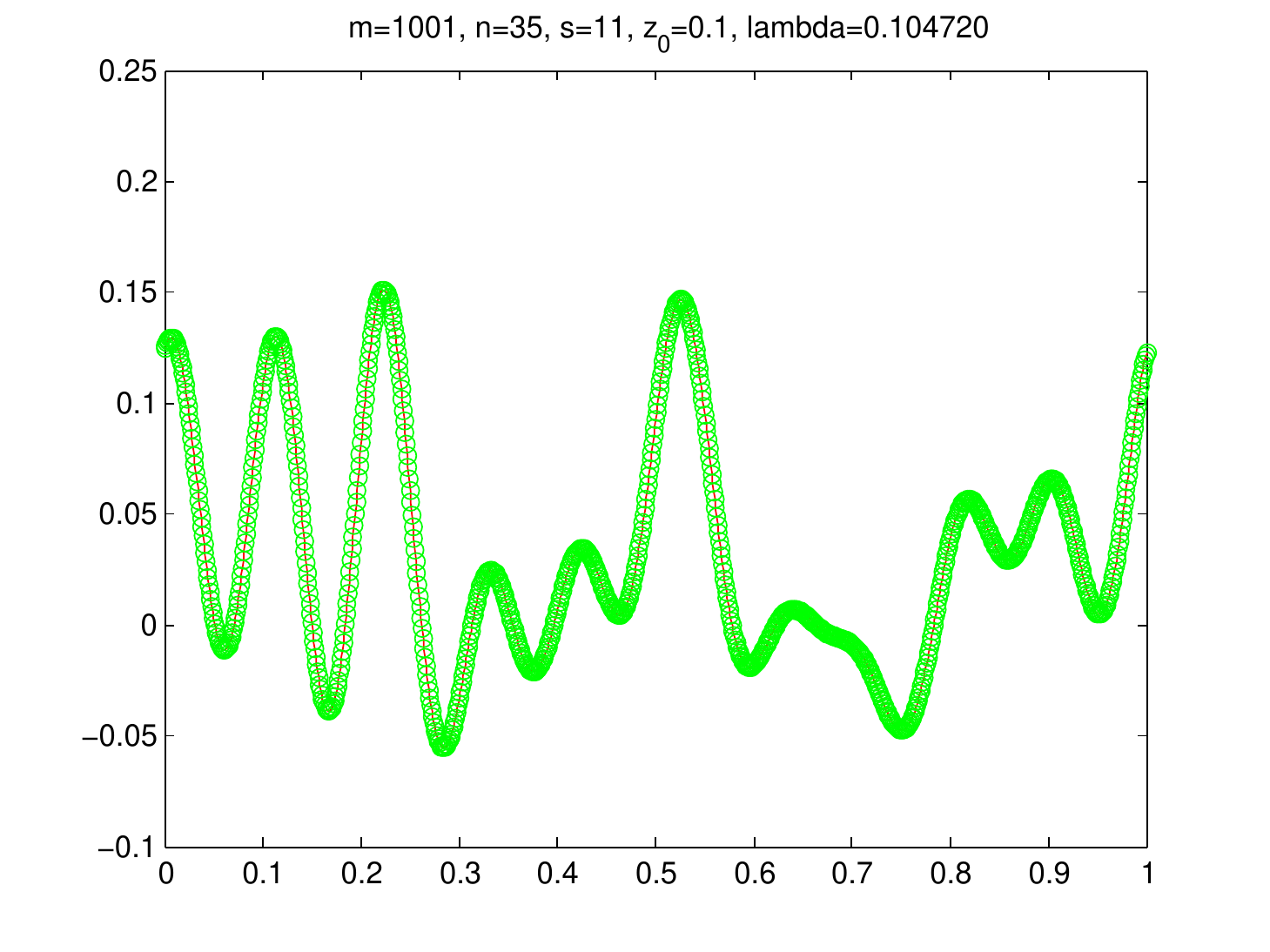}
\caption{In source inversion with far-field measurement (left, $z_0=1, n=35$),  the two  subwavelength modes $k=\pm 11$ cause significant  errors. 
With near-field measurement (right, $z_0=0.1, n=35$) the reconstruction is nearly perfect for target with 20 subwavelength
modes. The thin-red  curve
is the original profile and the thick-green curve
is the reconstructed profile. The two curves coincide in the plot
on the right.  }
\label{fig1}
\end{figure}

\begin{figure}[h]
\includegraphics[width=0.4\textwidth]{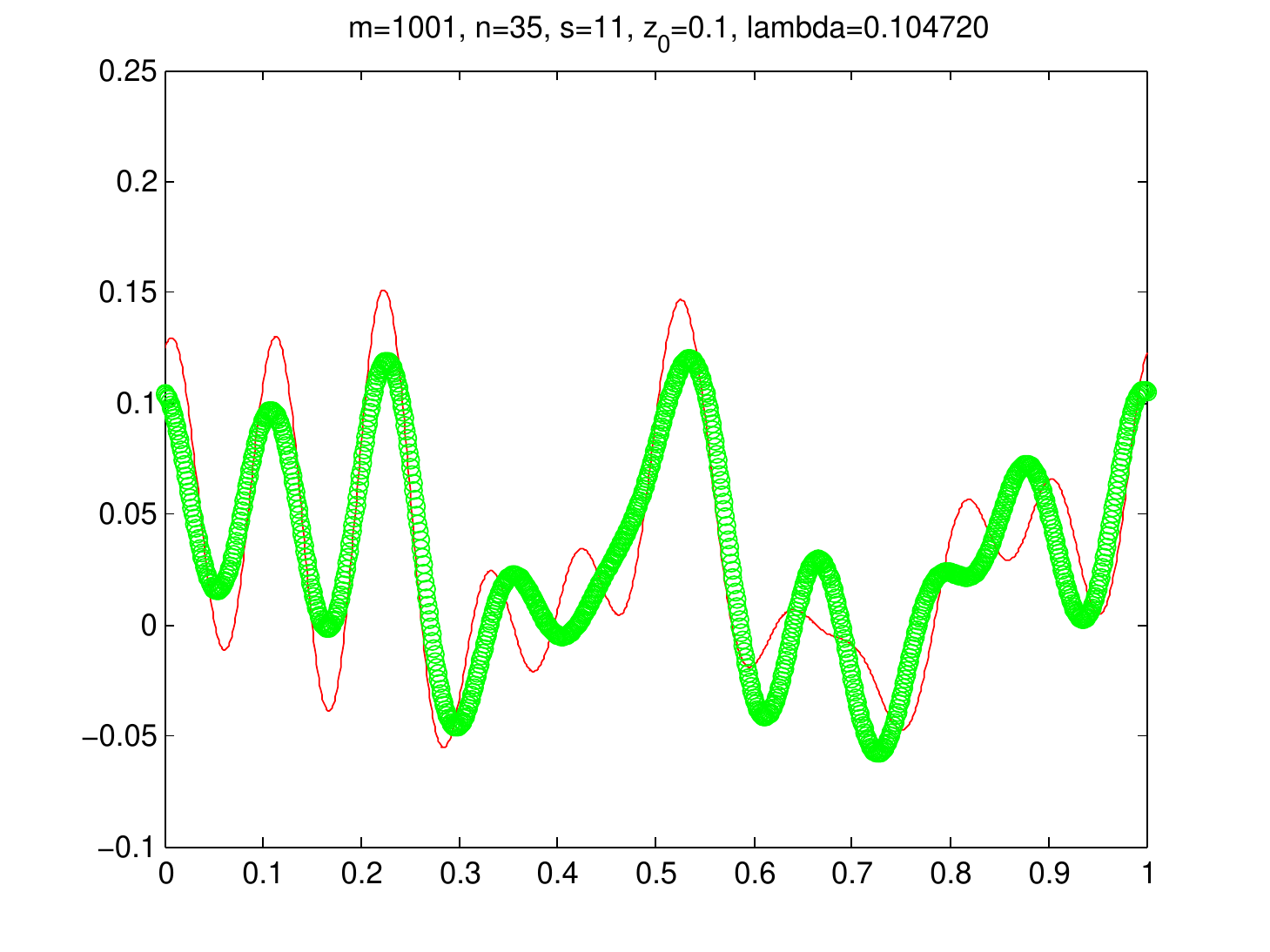}
\includegraphics[width=0.4\textwidth]{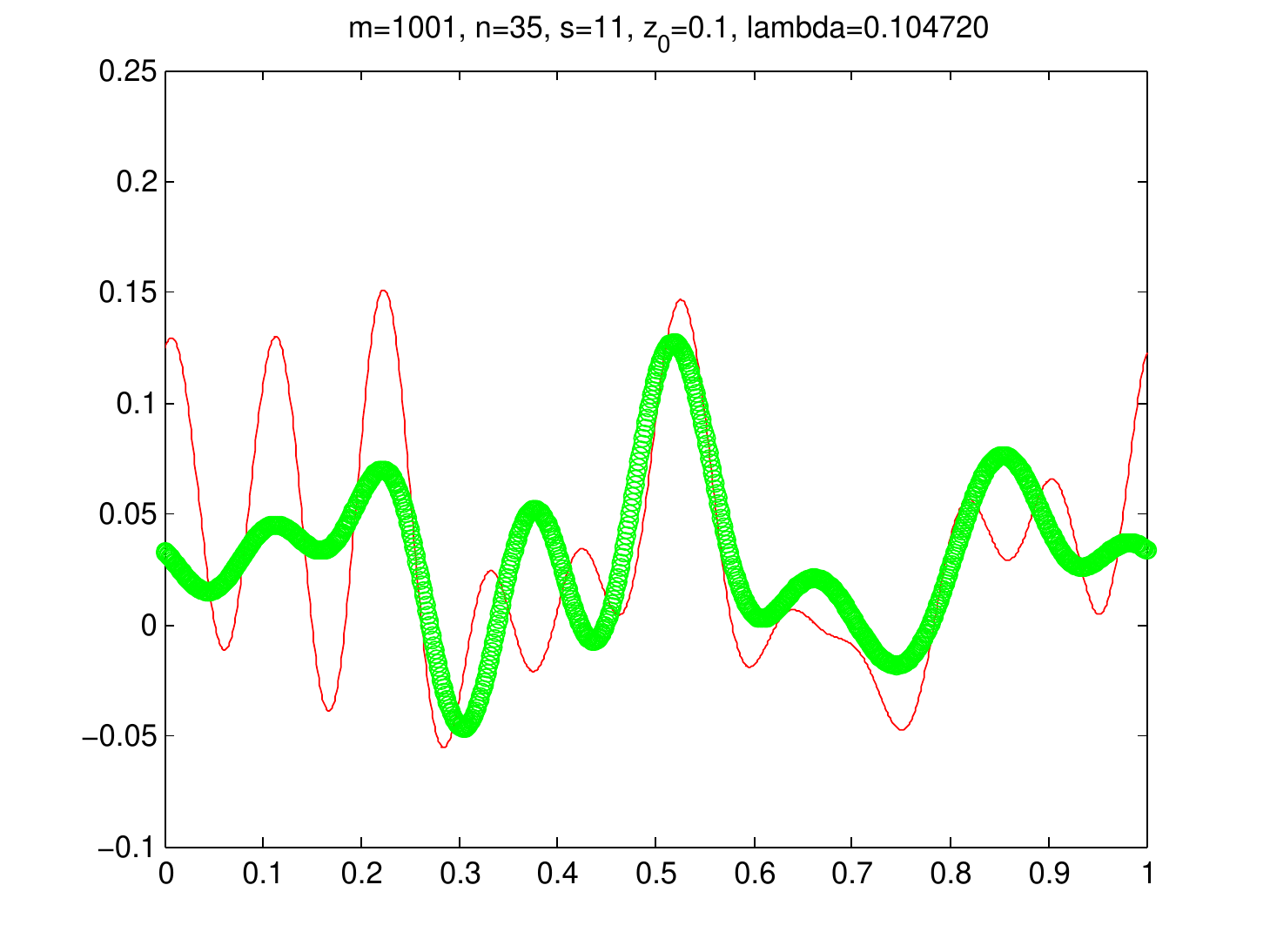}
\caption{In source inversion with near-field measurement with
added random noise (left $1\%$, right $5\%$).  The  thin-red curve
is the original profile and the thick-green   curve
is the reconstructed profile.  }
\label{fig1'}
\end{figure}

\begin{figure}
\includegraphics[width=0.4\textwidth]{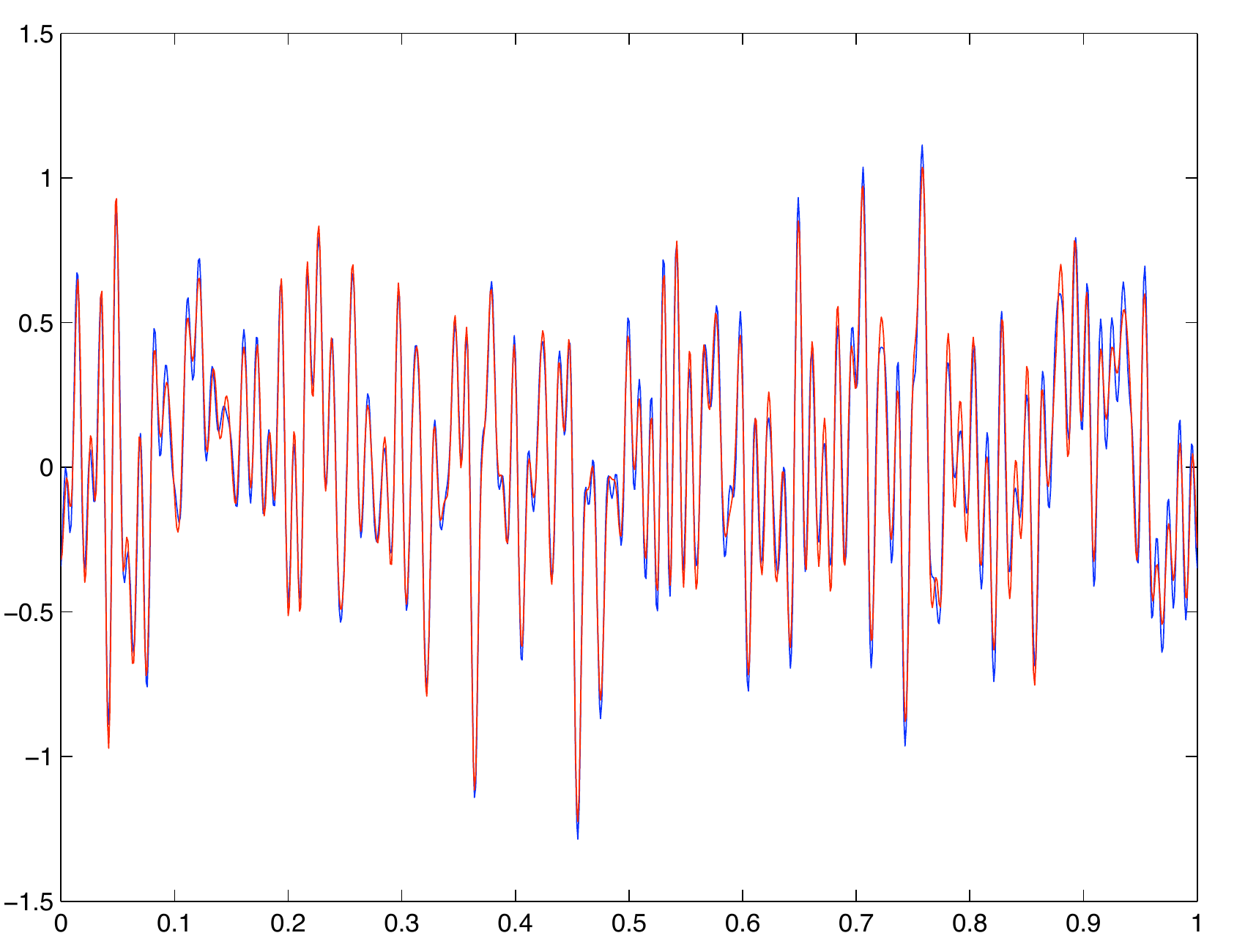}
\includegraphics[width=0.4\textwidth]{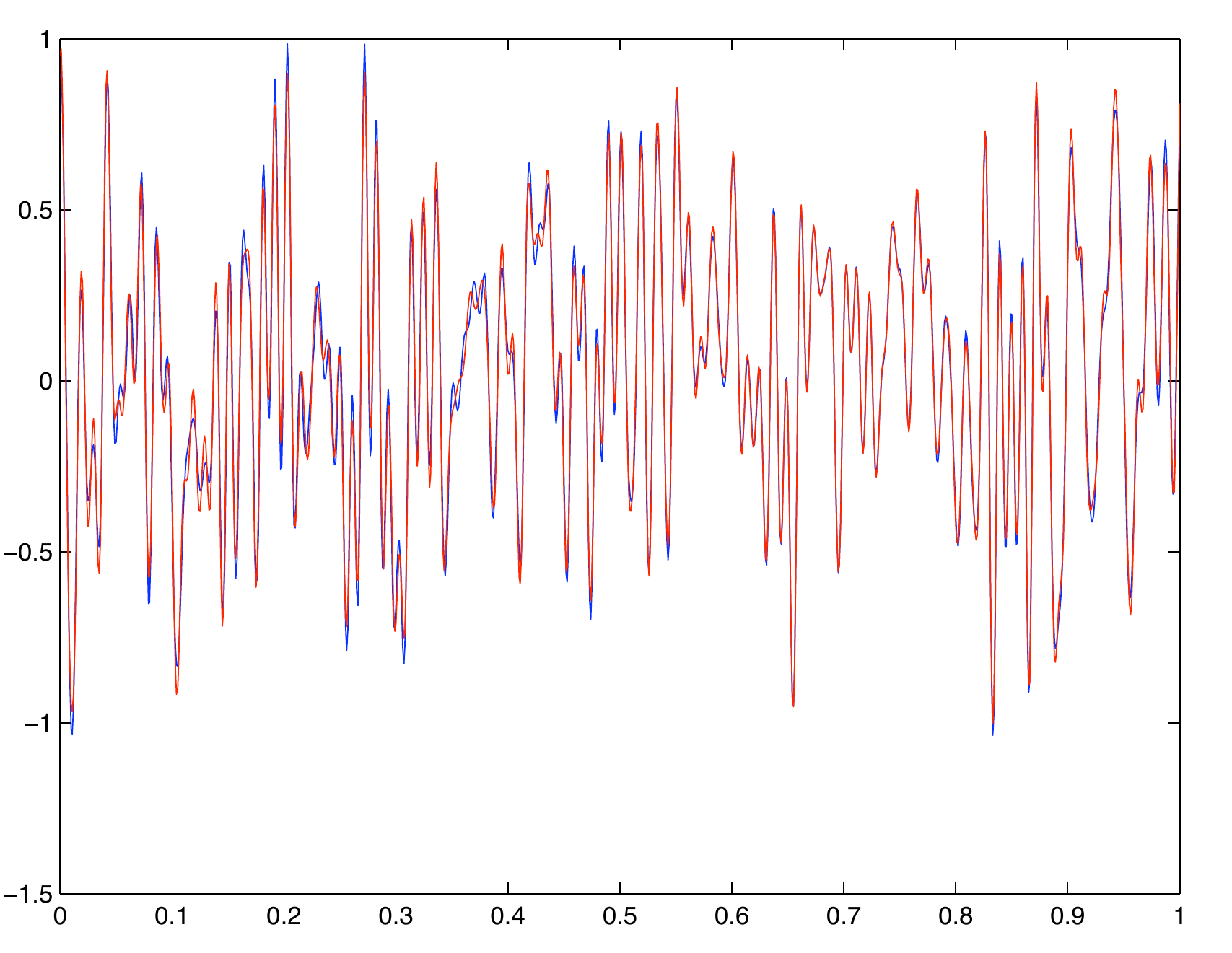}

\caption{Two independent runs (left with $n=70$ right with $n=82$)  with
far-field measurement ($z_0=1$) and  near-field illumination
($z_1-z_0=0.01$). The red curve is the original profile
 and the blue curve is the reconstructed profile. Except for
 a few spots, they coincide  with each other. }
\label{fig2}
\end{figure}
Numerical methods for convex optimization (such as
the Matlab program {\tt cvx} designed by M. Grant and S. Boyd) are not exact; they compute their results 
to within a predefined numerical precision or tolerance which is
acceptable for most applications. However, for 
our calculations involving subwavelength modes, 
the numerical error tends to be amplified enormously 
and hence severely spoils the reconstruction.

On the other hand, because of the simplicity of the
greedy algorithms, the numerical error can be much
reduced and the quality of recovery much improved. 
Of course, the sparsity constraints of the greedy algorithms
tend to be  more severe than  that of the convex relaxation.
Since the sparsity constraint is not the main focus 
of our study, we use the Subspace Pursuit in the numerical
calculations.

In our numerical simulations, $L=1$, $\lambda=\pi/30$, 
$m=1001$ so the subwavelength mode cutoff is at about $|k|=10$. The large $m$ is selected to demonstrate
the compressed sensing step. 

The numerical results for
source inversion are shown in Figure \ref{fig1}. 
First we image a periodic source with 11 modes,
including one  subwavelength mode $k= 11$
using far-field measurement ($z_0=1, n=35$). 
The main source of noise in 
this simulation is the roundoff error. 
As evident in the left plot of Figure \ref{fig1}
 the recovery is not accurate. Indeed,
 the subwavelength mode is entirely  missing
 in the reconstruction. 
Reducing the distance between the target and
the sensors ($z_0=0.1$) enables accurate
reconstruction  (right plot in Figure \ref{fig1}). 
When additional random noise is added, the
quality of recovery with near-field
measurement deteriorates in proportion to
the amount of noise (Figure \ref{fig1'}).
With about $5\%$ random noise, the
imaging result with near-field measurement (Figure \ref{fig1'}, right) 
is comparable  to that of  far-field measurement without
added noise (Figure \ref{fig1}, 
left).

 Figure \ref{fig2} shows the results
 for inverse Born scattering by the far-field
 measurement ($z_0=1$) and  the near-field
 illumination ($z_1-z_0=0.01$) without
 additional noise. 
 In this case it suffices to run the procedure
 for two incident modes with $q=\pm 90$. 
  
  \section{Extension and generalizations}
\label{sec4}
 \subsection{Non-periodic targets}\label{sec4.1}
 Let us consider a non-periodic target represented by
 a square-integrable function $\sigma$.
 
 To represent  $\sigma$
 we 
  use  the Littlewood-Paley basis
 \beq
 \label{lp}
 \hat \psi(\xi)=\lt\{
 \begin{matrix}
 (2\pi)^{-1/2}. & \pi\leq |\xi|\leq 2\pi\\
 0,&\hbox{otherwise}
 \end{matrix}
 \rt.
 \eeq
 or
 \beq
 \label{lp2}
 \psi(x)=(\pi x)^{-1} (\sin{(2\pi x)}-\sin{(\pi x)}). 
 \eeq
 Then the following set of functions 
 \beq
 \psi_{p,q}(x)=2^{-p/2}\psi(2^{-p}x-q),\quad p,q\in \IZ
 \eeq
 forms an orthonormal wavelet basis in $L^2(\IR)$ \cite{Dau}. 
  Expanding  the target profile $\sigma(x)$ in the Littlewood-Paley basis $\{\psi_{p,q}\}$  we have
 \beq
 \label{30}
 \sigma(x)=\sum_{p,q\in \IZ}\sigma_{p,q} \psi_{p,q}(x).
 \eeq

 \commentout{
The extended  target
 is represented by the strength function $\sigma(x)\in L^2(\IR)$
 \beq
 \label{0.1}
 \sigma(x)=\sum_{p,q\in \IZ}\tilde\sigma_{p,q}\psi_{p,q}(x).
 \eeq
 We assume that 
 \[
 \tilde\sigma_{p,q}=0,\quad |j|>M, \quad |k|>N
 \]
 for some known (usually large) integers $M$ and $N$.
 }
 \commentout{
 we can write  eq. (\ref{15}) as
 \beq
u^{\rm s}(z,x)&=&{i\over 4\pi}\sum_{p,q\in \IZ}\sigma_{p,q}
\int {d\alpha\over \beta} \int e^{i\om (\alpha x+\beta |z|)}
e^{-i\om\alpha x'} \psi_{p,q}(x')dx'
\eeq
}
\commentout{
Under the Born approximation equation becomes
\beq
u^{\rm s}(z,x)&=&{i\over 4\pi}\sum_{p,q\in \IZ}\tilde\sigma_{p,q}
\int {d\alpha\over \beta} \int e^{i\om (\alpha x+\beta |z|)}
e^{-i\om (\alpha-\alpha_0) x'} \psi_{p,q}(x')dx'\\
&=&{i\over \sqrt{8\pi}}\int {d\alpha\over \beta}\sum_{p,q\in \IZ}2^{p/2}\tilde\sigma_{p,q}
  e^{i\om (\alpha x+\beta |z|)}
e^{i\om (\alpha_0-\alpha) 2^{p} q} \hat \psi(\om (\alpha-\alpha_0) 2^{p}). \nn
\eeq
}

The scattered field has
the far-field asymptotic
\beq
u^{\rm s}(\br)={e^{i\om |\br|}\over |\br|^{1/2}}\lt(A
(\hat\br,\bd,\om)+\cO(|\br|^{-1})\rt),\quad \hat\br=\br/|\br|
\eeq
where  the scattering amplitude $A$ (at frequency $\om$ and
 the sampling direction $\hat\br$ with the incident field
 $u^{\rm i}=e^{i\bd\cdot\br}$)  can be calculated according to 
\beq
\label{sa}
A(\hat\br,\bd,\om)&=&{\om^2\over 4\pi}
\int d\br' \sigma(\br') e^{i\om \bd \cdot\br'} e^{-i\om \br'\cdot\hat\br},\quad \bd=(\alpha,\beta),\quad \alpha^2+\beta^2=1
\eeq
  \cite{CK}. 

In inverse scattering theory,  
 the  scattering amplitudes are measured in
 the directions $\hat\br_k=(\alpha_k,\beta_k)
 =(\cos\theta_k, \sin\theta_k),  k=1,...n$  and denoted by $Y=(Y_k)\in \IC^n$. 
By (\ref{sa}) and (\ref{30}) we have
\beq
\label{7}
Y_k={\om^2\over 2\sqrt{2\pi}} \sum_{p,q\in \IZ} 2^{p/2} \sigma_{p,q} e^{i\om (\alpha-\alpha_k)
2^{p} q}\hat\psi(\om(\alpha-\alpha_k) 2^{p}),\quad k=1,...,n.
\eeq

Let \beqn
l&=&\sum_{j=-p_*}^{p-1}(2m_j+1)+q,\quad |q|\leq m_p, \quad |p|\leq p_*, \\
k&=&\sum_{j=-p_*}^{p'-1}(2n_j+1)+q',\quad |q'|\leq n_p,
\quad |p'|\leq p_*
\eeqn
for some $m_p, n_p, p_* \in \IN$ be the column and row indices of the sensing matrix, respectively.  Set  $m=\sum_{|p|\leq p_*} (2m_p+1)$ and $n=\sum_{|p|\leq p_*}(2n_p+1)$.

Suppose that
\beq
\label{freq}
\om 2^{-p_*}\geq 2\pi
 \eeq
 and consider, for simplicity, the normal incident field
 $\alpha=0$. 
 Let $\zeta_{p',q'}$ be independent, uniform random variables on 
$[-1,1]$ and let
\beq
\alpha_k={\pi \over \om 2^{p'}}
\cdot\lt\{\begin{matrix}
1+\zeta_{p',q'},& \zeta_{p',q'}\in [0,1]\\
-1+\zeta_{p',q'},&\zeta_{p',q'}\in [-1,0]
\end{matrix}\rt.
\eeq
which lies in $[-1,1]$ by the assumption (\ref{freq}).
Let
 the sensing matrix elements be
\beq
\label{sense22}
\Phi_{k,l}= {1\over \sqrt{2n_p+1}}\hat\psi(\om\alpha_k 2^{p}) e^{-i\om\alpha_k
2^{p} q}. 
\eeq

We claim  that $\Phi_{k,l}=0$ for $p\neq p'$. This is evident
from the following calculation
\beq
\om\alpha_k 2^p={\pi 2^{p-p'}}
\cdot\lt\{\begin{matrix}
1+\zeta_{p',q'},& \zeta_{p',q'}\in [0,1]\\
-1+\zeta_{p',q'},&\zeta_{p',q'}\in [-1,0]
\end{matrix}\rt.
\eeq
whose right hand side is outside 
the support of $\hat \psi$ for $p\neq p'$. 

For every $p=p'$,  
\beq
\label{sense2}
\Phi_{k,l}= {1\over \sqrt{2n_p+1}} e^{-i \pi \zeta_{p',q'} q},\quad
|q'|\leq n_p,\quad |q|\leq m_p
\eeq
which constitute  the random partial Fourier matrix. 
In other words, under the assumption (\ref{freq}) 
the sensing matrix $\bPhi=[\Phi_{k,l}]\in \IC^{n\times m}$ is block-diagonal
with each block  (indexed by $p$) in the form (\ref{sense2}). 
\commentout{To enforce 
the constraint $\alpha_k\in [-1,1]$, we require
\[
\om 2^p\geq 2\pi
 \]
for all $p$ for which $\tau_{p,q}\neq 0$
and consider only 
\[
p' \geq \min{\{p: \tau_{p,q}\neq 0\}}.  
\]
}

\commentout{
 Moreover, in this case 
$\Phi_{k,l}\neq 0$ requires
 that
\[
\alpha_k={\xi\over \om 2^{p}} \in [-1,1],\quad |\xi|\in [\pi, 2\pi],
\]
in view of (\ref{lp}), which means that
\[
{\pi \over \om 2^{p}}< 1
\]
or equivalently
\[
2^{p}>\lambda/2.
\]
That is, only the information about the above-half-wavelength
scales is detected. 
}
Let  $X=(X_l)$ be the target vector  with
\[
X_l={\om^2\over 2\sqrt{2\pi}}\sqrt{(2n_{p}+1) 2^{p}} \sigma_{p,q},\quad l=\sum_{j=-p_*}^{p-1}(2m_j+1)+q. 
\]
We can then express the measurement vector $Y=(Y_k)\in \IC^n$  as in  (\ref{u1}).

Each block represented by the submatrix   (\ref{sense2}) can be treated as in
the periodic case. For each $p, |p|\leq p_*$, let $s_p$ be the sparsity of the
target vector for the scale $2^p$ and
suppose that the frequency $\om_p$ used to
probe the scale $2^p$ satisfies
\beq
\label{freq3}
\om_p2^p\geq 2\pi.
\eeq
Under these assumptions,   Theorem \ref{thm2} yields
the sufficient condition 
\beq
{2n_p+1\over \ln{(2n_p+1)}}\geq C \delta^{-2}s_p\ln^2{s_p} \ln{m} \ln{1\over \eta},\quad
\eta\in (0,1)
\eeq
for the RIP 
\[
\delta_{s_p}\leq \delta
\]
to hold with probability at least $1-\eta$. 

The assumption (\ref{freq3}) means that
the wavelength is at least as small
as the scale being  probed. 
Therefore,  this imaging method does not
possess the subwavelength resolving power.

Figure \ref{fig:LP} shows the result of 
reconstruction with the Littlewood-Paley basis
and the following parameters: $p_*=2, m_p=100,\forall p$;
for $ p=-2, -1, 0, 1,2$, 
$s_p=12, 24, 13, 24, 23, n_p=36, 64, 36, 64, 64,$ 
$\ep= 0.5134,     1.4849,     0.6520,     1.0274,     1.3681
$ equivalent to  the percentages  of noise $= 0.0500,     0.0649,     0.0494,     0.0640,     0.0698 $. 
The resulting reconstruction errors divided by the noises 
are $
    0.1941,     0.1240,     0.2498,     0.2361,     0.1502$,
    demonstrating the stability of the recovery.

\begin{figure}
\includegraphics[width=0.8\textwidth]{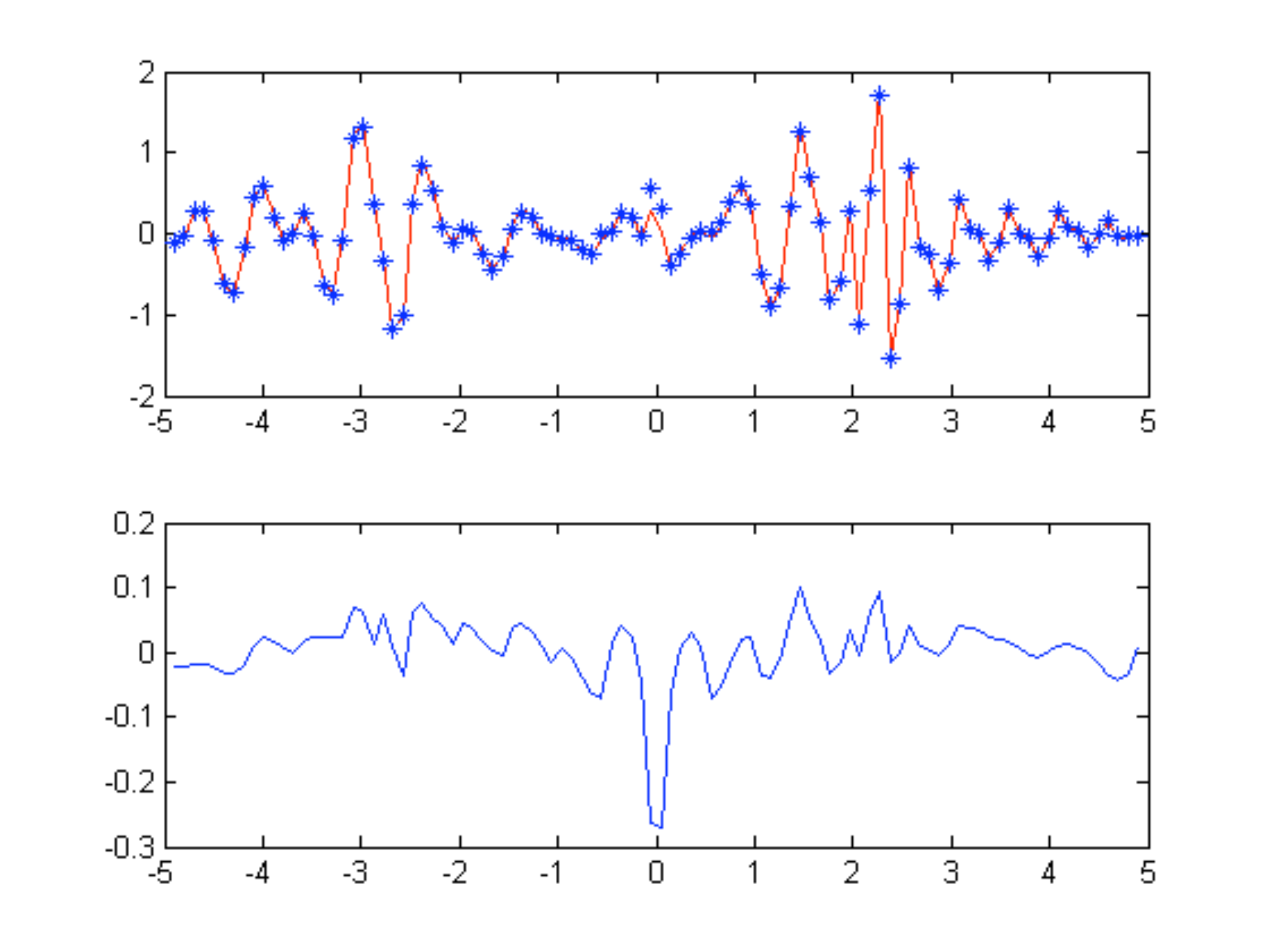}
\caption{ Imaging of  a non-periodic target in the presence of noise:  The red-solid curve (top)  is  the exact profile and the blue $*$ shows 
 the reconstructed  profile; the bottom plot shows the errors in reconstruction in the presence of noise. }
\label{fig:LP}
\end{figure}

\commentout{
\begin{figure}
\includegraphics[width=0.8\textwidth]{LP-figures/2.jpg}
\caption{Reconstruction  with the Littlewood-Paley basis ($p_*=2, m_p=100,\forall p$): (top) the exact profile, (middle) the reconstructed  profile, (bottom) the above two plotted together with the solid curve being the exact profile. }
\end{figure}
}

\commentout{
Consider the family of vectors $\{\Psi_j\}$ with
the components
\[
\Psi_{j,l}= {1\over \sqrt{N}}\hat\psi(\om\alpha_k 2^{p}) e^{-i\om\alpha_k
2^{p} q},\quad \alpha_k={\pi \over \om 2^{p}}(1+{q\over N}),
\]
where
\beqn
l&=&(M-m)N+q,\quad |p|\leq M, \quad n=1,...,N,\\
j&=&(M-p)N+q,\quad |p|\leq M,\quad q=1,...,N.
\eeqn
\begin{proposition}
The set of vectors  $\{\Psi_{j}\}$ forms  an orthonormal basis in $\IC^{N(2M+1)}$. 
\end{proposition}
\begin{proof}
First we check that this set is an orthogonal set.
For $j=(M-p)N+q, j'=(M-p)N+q$ we have
\beq
\sum_{p,q}\Psi_{j,l}\Psi^*_{j', l} &=&
{1\over {N}}\sum_{n=1}^N \hat\psi(\om\alpha_k 2^{p}) e^{-i\om\alpha_k
2^{p} q}\hat\psi^*(\om\alpha_{j'} 2^{m}) e^{i\om\alpha_{j'}
2^{p} q}. \label{8}
\eeq
Since 
$\hat\psi(\om\alpha_k 2^{p}) $ is zero unless
$p=p'$, (\ref{8}) vanishes unless $p=p'$. 
In this case, (\ref{8}) reduces to 
\[
{1\over {N}}\sum_{n=1}^N e^{-i\pi(1+q/N) n} e^{i\pi(1+q/N) n}
\]
which vanishes for $q\neq q$.
When $q=q, p=p'$ it is easily verified that (\ref{8}) equals unity.
 \end{proof}
For each $i=1,...,(2M+1)$,  let $I_{k_i}, k_i=1,...,N$ 
be independent Bernoulli random variables taking the value
1 with probability $\bar K/N, \bar K\ll N.$  Let 
$
K_i $ be the cardinality of the random set
\beq
\Omega_i= \{ k_i\in \{1,...,N\}| I_{k_i}=1\},\quad i=1,...,(2M+1)
\eeq
and let $k_i^*=1,...,K_i$ enumerate  the members of $\Omega_i$. We have  $\IE|\Omega_i|=\bar K, \forall i$ and indeed by 
Bernstein's inequality $ K_i, \forall i$ are  close to $\bar K$
with high probability if $N$ is sufficiently large. 
For every $m$ and $n$ define 
  \beq
 \label{10}
 \Phi_{k, l}=\sqrt{N\over K_m} \Psi_{k^*_m,l}, \quad  k_m^*=1,...,K_m, \quad l=(M-m)N+q
 \eeq
for $ k=k_m^*+\sum_{j=1}^{i-1} K_j$ 
and $\Phi_{k,l}=0$ otherwise.

Let  $\bPhi=[\Phi_{k, l}]$ be the sensing matrix and
let  $X=(X_l)$ with
\[
X_l=\sqrt{2\pi K_m 2^{p}} \sigma_{p,q},\quad l=(M-m)N+q 
\]
be the target vector. 
We can then express the datum vector $Y=(Y_k)$  as 
\beq
\label{u1}
Y=\bPhi X
\eeq

}

\commentout{
\subsection{Non-periodic band-limited  targets} Consider
the source inversion problem and  
suppose that  the target  function $\sigma(x)$ is band-limited 
and  admits  the Fourier integral
\[
\sigma(x)=\int \hat \sigma(k) e^{i2\pi kx} dk
\]
where ${\rm supp} (\hat\sigma)=[a,b]$. 
Then
\beq\label{31}
\int G(z_0, x-\xi) \sigma(x) dx =-{i\over 2\omega}
\int^b_a dk {\hat\sigma(k)\over \beta(k\lambda)} e^{i\omega k \lambda \xi}
e^{i\omega \beta(k\lambda)  z_0}. 
\eeq
 Consider 
 the Trapezoidal rule  for the integral (\ref{31}) 
 with step size $h=(b-a)/(m-1)$ 
\beq
\label{22}
{-i\over \omega }\lt[ {\hat\sigma_1\over 2\beta_1}e^{i\om\beta_1 z_0}e^{i\om\alpha_1\xi}+{\hat\sigma_m\over 2\beta_m} 
 e^{i\om\beta_m z_0} e^{i\om\alpha_m\xi}
+\sum_{k=2}^{m-1} {\hat\sigma_k\over \beta_k} 
e^{i\omega \beta_k z_0}e^{i\omega \alpha_k \xi}\rt]+E
 \eeq
 with
 \beqn
 \hat \sigma_k &=&h\hat\sigma(a+kh)\\
 \alpha_k&=&(a+ kh) \lambda\\
 \beta_k&=& \beta(\alpha_k)
 \eeqn
 where the error $E$ is given by
 \[
 E=-{(b-a)^3\over 12 (m-1)^2} f''(\zeta),\quad
 f(k)=-{i\over 2\omega}
{\hat\sigma(k)\over \beta(k\lambda)} e^{i\omega k \lambda \xi}
e^{i\omega \beta(k\lambda)  z_0}
 \]
 for some $\zeta\in [a,b]$. Hence, $\epsilon=\|E\|_2=\cO(h^2)$. 
 
 Let 
 \beq
 X_1&=&-{i\sqrt{p}\over \omega } {\hat\sigma_1\over 4\beta_1}e^{i\om\beta_1 z_0}\hat\sigma_1\\
 X_k&=&-{i\sqrt{p}\over \omega }{\hat\sigma_k\over 2\beta_k} 
e^{i\omega \beta_k z_0}\hat\sigma_k,\quad k=2,...,m-1\\
 X_m&=&-{i\sqrt{p}\over \omega }{\hat\sigma_m\over 4\beta_m}e^{i\om\beta_m z_0}\hat\sigma_m. 
 \eeq
Then the corresponding solution $X^\ep$ of 
 (\ref{32}) differs from the true target vector $X$ by
 $\cO(\ep)=\cO(h^2)$. 
}
 
 \subsection{Three dimensions}\label{sec4.2}
 
 The present framework  can be easily  extended to three dimensions by using the plane-wave   
representation 
\beq
{e^{i\om|\br|}\over |\br|}={i\om\over 2\pi}
\int {d\alpha d\beta \over \gamma} 
 \exp{\lt[i\om(\alpha x+\beta y+
 \gamma|z|)\rt]}, \quad \br=(x,y,z)
\label{weyl}
 \eeq
 where
 \beqn
\gamma&=&\sqrt{1-\alpha^2-\beta^2},\quad \alpha^2+\beta^2\leq 1\\
\gamma&=&i\sqrt{\alpha^2+\beta^2-1},\quad \alpha^2+\beta^2>1
\eeqn
for the Green function 
\cite{BW}. 

 \section{Conclusion}\label{sec6}
In this note, we have analyzed the problem of imaging 
extended  targets  in the perspective of Compressed Sensing which provides assurance of stable  
reconstruction of a target composed of sparse Fourier  modes
with the similar number of measurements, modulo
a poly-logarithmic factor, 
 in the presence of noise. 
 
 We have shown that the number of
 stably recoverable modes grows as the negative
 $d$-th power of the distance between the
 target and the sensors/source. Hence
 the stability of reconstructing subwavelength modes 
 requires the distance to be less than the 
 wavelength. On the other hand, we have also shown that
 the resolution limit is inversely proportional to
the  SNR in the high SNR limit. As a consequence,
the subwavelength modes
can be recovered at sufficiently high SNR as well
as by placing the sensors at a subwavelength distance
from the target.  

However, it remains to be seen if
 these results about stability and resolution can
 be extended  to the case of square-integrable targets even though  our sampling scheme block-diagonalizes 
 the corresponding scattering matrix according to 
 each dyadic scale.  
     
Finally we note that the compressive imaging theory in the remote sensing regime
for discrete point targets has been recently developed
in  \cite{cis-simo, cis-siso, CS-par}.

\bigskip
{\bf Acknowledgement.} I thank my students Arcade Tseng
(Figures \ref{fig1}-\ref{fig2}) and Wenjing Liao (Figure \ref{fig:LP})  for 
 preparing  the figures. I am grateful to the anonymous referee
 for the helpful comments for improving the manuscript.

\end{document}